\documentclass{elsarticle-arxiv}

\usepackage[english]{babel}
\usepackage[utf8]{inputenc}
\usepackage[T1]{fontenc} 

\usepackage{amsfonts}
\usepackage{amsmath}
\usepackage{caption}
\usepackage{comment}
\usepackage{environ}
\usepackage{scrlayer-scrpage}
\usepackage{float}
\usepackage{fontaxes}
\usepackage{graphics}
\usepackage{hyperref}
\usepackage{manyfoot}
\usepackage{microtype}
\usepackage{mweights}
\usepackage{natbib}
\usepackage{nccfoots}
\usepackage{setspace}
\usepackage{totpages}
\usepackage{trimspaces}
\usepackage{upquote}
\usepackage{url}
\usepackage{xcolor}
\usepackage{xkeyval}
\usepackage{mathtools}
\usepackage{amssymb}

\definecolor{dark-blue}{rgb}{0.05,0.25,0.85}
\hypersetup{colorlinks=true,linkcolor=dark-blue,citecolor=dark-blue}

\usepackage[amsmath,thmmarks,hyperref]{ntheorem}
\usepackage[nameinlink]{cleveref}

\crefformat{page}{#2page~#1#3}%
\Crefformat{page}{#2Page~#1#3}%
\crefformat{equation}{#2(#1)#3}%
\Crefformat{equation}{#2(#1)#3}%
\crefformat{figure}{#2Figure~#1#3}%
\Crefformat{figure}{#2Figure~#1#3}%
\crefformat{section}{#2Section~#1#3}
\Crefformat{section}{#2Section~#1#3}
\crefformat{chapter}{#2Chapter~#1#3}
\Crefformat{chapter}{#2Chapter~#1#3}
\crefformat{chapter*}{#2Chapter~#1#3}
\Crefformat{chapter*}{#2Chapter~#1#3}
\crefformat{part}{#2Part~#1#3}
\Crefformat{part}{#2Part~#1#3}
\crefformat{enumi}{#2(#1)#3}
\Crefformat{enumi}{#2(#1)#3}
\crefformat{algorithm}{#2algorithm~#1#3}%
\Crefformat{algorithm}{#2Algorithm~#1#3}%

\usepackage{latexsym}

\theoremnumbering{arabic}
\theoremstyle{plain}
\theoremsymbol{}
\theorembodyfont{\itshape}
\theoremheaderfont{\normalfont\bfseries}
\theoremseparator{}

\theoremseparator{.}
\newtheorem{theorem}{Theorem}
\crefformat{theorem}{#2Theorem~#1#3}
\Crefformat{theorem}{#2Theorem~#1#3}

\newcommand{\newtheoremwithcrefformat}[2]{%
  \newtheorem{#1}[theorem]{#2}%
  \crefformat{#1}{##2\MakeUppercase#1~##1##3}%
  \Crefformat{#1}{##2\MakeUppercase#1~##1##3}%
}

\newtheoremwithcrefformat{proposition}{Proposition}
\newtheoremwithcrefformat{observation}{Observation}
\newtheoremwithcrefformat{lemma}{Lemma}
\newtheoremwithcrefformat{conjecture}{Conjecture}
\newtheoremwithcrefformat{corollary}{Corollary}
\theorembodyfont{\upshape}
\newtheoremwithcrefformat{example}{Example}
\newtheoremwithcrefformat{remark}{Remark}

\newtheoremwithcrefformat{definition}{Definition}

\theoremsymbol{\ensuremath{\square}}
\theoremstyle{nonumberplain}
\theoremseparator{}
\theoremheaderfont{\scshape}
\theorembodyfont{\normalfont}
\theoremsymbol{\ensuremath{\square}}
\newtheorem{proof}{Proof.}

\newcommand{\sref}[2]{\hyperref[#2]{#1~\ref{#2}}}

\newcommand{\Oof}{\mathcal{O}}

\newcommand{\Gg}{\mathcal{G}}
\newcommand{\Ff}{\mathcal{F}}

\newcommand{\N}{\mathbb{N}}

\usepackage{tikz}
\usetikzlibrary{decorations.pathreplacing}
\tikzstyle{vertex}=[circle,inner sep=1.5,minimum size =1.5mm,semithick,fill=black, draw=black]

\title{Greedy domination on biclique-free graphs}
\author{Sebastian Siebertz}
\address{Humboldt-Universit\"at zu Berlin}
\ead{siebertz@informatik.hu-berlin.de}
\fntext[fn0]{Full address: Sebastian Siebertz, Humboldt-Universit\"at zu Berlin, Institut f\"ur Informatik, Chair for Algorithm Engeneering, email: \texttt{sebastian.siebertz@hu-berlin.de}}
\fntext[fn1]{This work 
was supported by the National Science Centre of 
Poland via POLONEZ grant agreement UMO-2015/19/P/ST6/03998, 
which has received funding from the European Union's Horizon 2020 research and innovation programme (Marie Sk\l odowska-Curie grant agreement No.\ 665778).}

\begin{document}

\begin{frontmatter}
\begin{abstract}
The greedy algorithm for approximating dominating
sets is a simple method that is known to compute a factor
$(\ln n+1)$ approximation of a minimum dominating set on 
any graph with $n$ vertices. We show that a small modification of the 
greedy algorithm can be used to compute a factor $\Oof(t\cdot \ln k)$ approximation, where~$k$ is the size of a minimum dominating
set, on graphs that exclude the complete bipartite graph $K_{t,t}$
as a subgraph. 
\end{abstract}
\end{frontmatter}
\begin{picture}(0,0) \put(310,-303)
{\hbox{\includegraphics[scale=0.2]{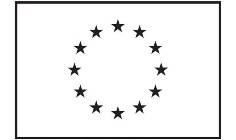}}} \end{picture} 
\vspace{-0.8cm}
\section{Introduction}

A dominating set in an undirected and simple
graph $G$ is a set $D\subseteq V(G)$
such that every vertex $v\in V(G)$ lies either in $D$ or has a neighbour
in $D$. The \textsc{Minimum Dominating Set} problem
takes a graph $G$ as input and the objective is to find a 
minimum size dominating set of~$G$. 
The corresponding decision problem is 
NP-hard~\cite{karp1972reducibility} and
this even holds in very restricted settings, e.g.\ on planar graphs of
maximum degree~$3$~\cite{garey2002computers}.

The following greedy algorithm approximates \textsc{Minimum 
Dominating Set} in an $n$-vertex graph $G$ up to 
a factor 
$H_n=\sum_{i=1}^n1/i\leq (\ln n+1)$~\cite{johnson1974approximation,lovasz1975ratio}. 
Starting with the empty dominating set~$D$, the
algorithm iteratively adds vertices to~$D$ according to the
following greedy rule until all vertices are dominated: in each round, choose the vertex
\mbox{$v\in V(G)$} that dominates the largest number of vertices
which still need to be dominated.
The greedy algorithm on general graphs is almost optimal: 
it is NP-hard to approximate \textsc{Minimum Dominating Set}
within factor $c\cdot \ln n$ for some constant~$c>0$~\cite{raz1997sub}, and by a recent result it is even NP-hard to 
approximate \textsc{Minimum Dominating Set} within factor 
$(1-\epsilon)\cdot \ln n$ for every~$\epsilon>0$~\cite{dinur2014analytical}.

On several restricted graph classes \textsc{Minimum Dominating Set} can 
be approximated much better. For instance, the problem 
admits a polynomial-time approximation scheme (PTAS) on planar
graphs~\cite{baker1994approximation} and, more generally, on graph
classes with subexponential expansion~\cite{har2017approximation}. 
It admits a constant factor approximation on classes of bounded 
arboricity~\cite{bansal2017tight} and an $\Oof(d\cdot \ln k)$ approximation (where $k$ denotes the size of a minimum dominating
set) on classes of VC-dimension~$d$~\cite{bronnimann1995almost,even2005hitting}. 
While the above algorithms on restricted graph classes yield
good approximations, they are computationally much more complex
than the greedy algorithm. Unfortunately, the greedy algorithm
does not provide any better approximation on these restricted
graph classes than on general graphs (see for example 
Section~4 of~\cite{bronnimann1995almost} for an instance of the set
cover problem, which can easily be transformed into a planar
instance of the dominating set problem, where the greedy algorithm 
achieves only an $\Omega(\ln n)$ approximation). Jones et 
al.~\cite{jones2013parameterized} showed how to slightly change
the classical greedy algorithm to obtain a constant factor
approximation algorithm on sparse graphs, more precisely, 
the algorithm computes a $d^2$ approximation of 
\textsc{Minimum Dominating Set} 
on any graph of degeneracy at most $d$.

\paragraph{Our results}

We follow the approach of Jones et al.~\cite{jones2013parameterized}
and study small modifications of the greedy algorithm
which lead to improved approximations on restricted graph
classes. 
We denote the complete bipartite graph with~$i$ vertices 
on one side and $j$ vertices on the other side by $K_{i,j}$. 

We
present a greedy algorithm which takes as input a graph $G$
and an optional
parameter $i\in \N$. If run with the integer 
parameter $i$ and $G$ excludes $K_{i,j}$ as a subgraph
for some $j\geq 1$, then the algorithm computes an 
$\Oof\big(i^2\cdot \ln k+ i\cdot \ln j\big)$ approximation of
\textsc{Minimum Dominating Set} (where $k$ denotes the size
of a minimum dominating set in $G$). 
%
If run without the integer parameter
$i$, the algorithm outputs the largest subgraph~$K_{t,t}$ that it
found during its computation, as well as an 
$\Oof(t^2\cdot \ln k+$ $t\cdot \ln t)= \Oof(t^2\cdot \ln k)$ approximation of \textsc{Minimum Dominating Set}. 
By running the classical greedy algorithm in parallel, the approximation
ratios can be improved to $\Oof(i \cdot \ln k+\ln j)$, and 
$\Oof(t\cdot \ln k)$, respectively. 

Based on a known hardness result for the set cover problem
on families with intersection $1$ it is easy to show that 
it is unlikely that polynomial time constant 
factor approximations exist even on $K_{3,3}$-free graphs.

\paragraph{Comparison to other algorithms}

Every $K_{t,t}$-free graph has VC-dimension at most~$t$, hence
the algorithms of~\cite{bronnimann1995almost,even2005hitting}
achieve $\Oof(t\cdot \ln k)$ approximations on the graphs we
consider. The algorithm presented
in~\cite{bronnimann1995almost} is based on finding $\epsilon$-nets
with respect to a weight function and a polynomial number of 
reweighting steps. The algorithm presented in~\cite{even2005hitting}
requires solving a linear program. Hence, even though these
algorithms achieve the same approximation bounds as 
our modified greedy algorithm, 
our algorithm is much easier to implement and has much better
running times. On the other hand, $K_{t,t}$-free-graphs are
strictly more general than degenerate graphs. Hence, our algorithm
is applicable to a more general class of graphs than the algorithm
of Jones et al.~\cite{jones2013parameterized}. 



\section{The greedy algorithm on biclique-free graphs}

We first consider the following greedy algorithm which takes as 
input an optional parameter $i\in \N$ and a graph $G$. We start by
presenting how the algorithm works if the parameter $i$ is given 
with the input. 


We initialise $D_0\coloneqq \emptyset$ and $A_0\coloneqq V(G)$. 
The set $D_0$ denotes the initial dominating set and $A_0$ 
denotes the set of vertices that have to be dominated. The algorithm
runs in rounds and in every round it makes a greedy choice 
on a few vertices to add to the dominating set, 
until no vertices remain to be dominated. 
Formally, in each round $m=1,\ldots,$ we construct a new set
$D_m$ which is obtained from $D_{m-1}$ by adding at 
most~$i-1$ vertices $v_1,\ldots,v_\ell$. 
The set $A_m$ is obtained from $A_{m-1}$ by
removing $v_1,\ldots,v_\ell$ and their neighbours. 
We output the set $D_m$ as a dominating set, 
when $A_m=\emptyset$. 

Let us describe a round of the modified greedy algorithm. 
Assume that after round~$m$ we have constructed a partial dominating
set $D_m$ and vertices $A_m$ remain to be dominated. We 
choose $\ell$ vertices $v_1,\ldots, v_\ell$, $\ell< i$, as follows. 
We choose as~$v_1$ an arbitrary vertex that dominates the
largest number of vertices which still need to be dominated, 
i.e., a vertex which maximises $|N[v_1]\cap A_m|$. Here, 
$N[v]$ denotes the neighbourhood of a vertex~$v$, including
the vertex $v$. Let $B_1\coloneqq (N[v_1]\cap A_m)\setminus \{v_1\}$. 
We continue to choose vertices $v_2,\ldots, v_\ell$ inductively
as follows. If the vertices $v_1,\ldots,v_s$ and sets 
$B_1,\ldots, B_s\subseteq V(G)$ have been defined, we choose 
the next vertex $v_{s+1}$ as an arbitrary vertex not in $\{v_1,\ldots,
v_s\}$ that dominates
the largest number of vertices of $B_s$, i.e., a vertex which
maximises $|N[v_{s+1}]\cap B_s|$ and let 
$B_{s+1}\coloneqq (N[v_{s+1}]\cap B_s)\setminus \{v_{s+1}\}$. 
We terminate this round and add $v_1,\ldots, v_\ell$ to $D_{m+1}$ if 
either we have $\ell=i-1$, or $N[v]\cap B_\ell=\emptyset$ for each 
$v\in V(G)\setminus\{v_1,\ldots, v_\ell\}$. We mark the vertices $v_1,\ldots, v_\ell$ and their neighbours as dominated, i.e., we remove
from the set $A_m$ all vertices of $\bigcup_{1\leq m\leq \ell} N[v_m]$
to obtain the set $A_{m+1}$ and start the next round. 

\smallskip
The crucial difference between the above modified greedy algorithm and the 
classical greedy algorithm is that the former is guaranteed to choose in 
every round $m$ at least one vertex from \emph{every} 
minimum dominating set for $A_m$, given that $A_m$ is still large. 
This is made precise in the following lemma.  

\begin{lemma}\label{lem:A-large}
Let $G$ be a graph which excludes $K_{i,j}$ as a subgraph. Let 
$A_m\subseteq V(G)$ be a set of vertices to be dominated 
and let $M$ be a dominating set of $A_m$ of size~$k$ in $G$. If 
$|A_m|\geq k^{i}\cdot (j+i)$, 
then the algorithm applied to $A_m$ will find vertices 
$v_1,\ldots, v_\ell$ with $M\cap \{v_1,\ldots, v_\ell\}\neq
\emptyset$. 
\end{lemma}
\begin{proof}
By assumption, $A_m$ is dominated by the set $M$ of size $k$. Hence
there must exist a vertex $v_1\in V(G)$ which dominates at least
a $1/k$ fraction of $A_m$, that is, at least $k^{i-1}\cdot (j+i)$ vertices
of $A_m$. Let $B_1\coloneqq (N[v_1]\cap A_m)\setminus\{v_1\}$, 
hence $|B_1|\geq k^{i-1}\cdot (j+i)-1\geq k^{i-1}\cdot (j+i-1)$. 

Assume $v_1\not\in M$. We repeat the same argument
as above for $B_1$. Also $B_1$ is dominated by $M$ of size $k$, hence 
there must exist a vertex $v_2\in V(G)$ which
dominates at least a $1/k$ fraction of $B_1$, that is, at least 
$k^{i-2}\cdot (j+i-1)$ vertices of~$B_1$. Let $B_2\coloneqq 
(N[v_2]\cap B_1)\setminus\{v_2\}$, 
hence $|B_2|\geq k^{i-2}\cdot (j+i-1)-1\geq k^{i-2}\cdot (j+i-2)$. 
We repeat the argument for $v_2,v_3,\ldots,v_{\ell}$ and $B_2,
B_3,\ldots, B_{\ell}$, each $B_x$ 
for $0\leq x\leq \ell<i$ (set $B_0=A_m$) of size at 
least $k^{i-x}\cdot (j+i-x)$, ending with a set~$B_{\ell}$ of 
size at least $k\cdot (j+1)$. 

Hence, assuming that $v_1,\ldots, v_\ell\not\in M$, we have $\ell=i-1$
and there must 
exist a vertex~$v$ with 
$|N[v]\cap B_\ell\setminus \{v\}|\geq j$. Fix any subset $B=\{w_1,\ldots, w_j\}$ 
of $N[v]\cap B_\ell\setminus \{v\}$ of size exactly $j$. Then the vertices
$v,v_1,\ldots, v_{i-1}$ and the vertices $w_1,\ldots, w_j$
form a subgraph $K_{i,j}$, contradicting that
such a subgraph does not exist in $G$. Hence, one of 
$v_1,\ldots, v_\ell$ must be contained in $M$. 
\end{proof}

Hence, as long as it remains to dominate a large set $A_m$, 
the modified greedy algorithm makes an almost optimal choice. 
Once we are left with a small set~$A_m$, it performs only slightly
worse than the classical greedy algorithm. 

\begin{theorem}
If $G$ is a graph which excludes $K_{i,j}$ as a subgraph, then
the modified greedy algorithm called with parameter $i$ 
computes an $\Oof(i^2\cdot \ln k+ i\cdot \ln j)$ approximation of 
a minimum dominating set of $G$, where $k$ is the size of a 
minimum dominating set of $G$.
\end{theorem}
\begin{proof}
Fix any minimum dominating set $M$ of size $k$ of $G$. 
By \cref{lem:A-large}, as long as it remains to dominate
a set of size at least $k^i\cdot (j+i)$, the modified greedy
algorithm chooses in every round at least one vertex of $M$. 
Hence, when it remains to dominate a set $A_m$ of size
smaller than $k^i\cdot (j+i)$, the algorithm has chosen
at most $i\cdot k$ vertices. 

Once we have reached this situation, let $n\coloneqq |A_m|
\leq k^i\cdot (j+i)$. We argue just as  
in the proof of \cref{lem:A-large} that there exists a vertex 
$v\in V(G)$ which dominates at least a $1/k$ fraction of 
$A_m$, that is, a subset of $A_m$ of size at least~$n/k$. The algorithm
chooses such a vertex together with at most $i$ other vertices 
which in the worst case dominate nothing else. Hence after
the first round we are 
left to dominate at most $n_1=n-n/k=n\cdot (1-1/k)$
vertices. In the second round, we find again a vertex which dominates
at least a $1/k$ fraction of the remaining vertices, hence
after the second round we are left to dominate at most 
$n_2=n_1-n_1/k=n_1\cdot (1-1/k)=n\cdot (1-1/k)^2$ vertices. 
We repeat this argumentation and conclude that after 
executing $x$ rounds of the algorithm it remains to dominate
at most $n_{x}=n\cdot (1-\frac{1}{k})^{x}$ elements.
Let us determine for what value of $x$ we have $n_x < 1$, in
which case we have dominated all vertices. 

We have $n_x\leq n\cdot (1-1/k)^x<n\cdot e^{-x/k}$, where the last inequality follows from the bound $1-z<e^{-z}$, which holds
for all $z>0$. Thus, for $x\geq k\cdot \ln n$ we have $n_x<n\cdot e^{-\ln n}=1$. 
We conclude that the algorithm terminates after at most
$k\cdot \ln n$ steps, in particular, it computes a dominating set of
size at most~$i\cdot k\cdot \ln n$. 
Now, as $n\leq k^i\cdot (j+i)$, 
we have $\ln n\in \Oof(i\cdot \ln k+\ln j)$. Hence, in total 
the set has size at most 
$\Oof\big(i\cdot k + (i^2\cdot \ln k+i\cdot \ln j)\cdot k\big)
\in \Oof\big((i^2\cdot \ln k+i\cdot \ln j)\cdot k\big)$. 
\end{proof}


With slightly more computational effort 
we can compute an $\Oof(i\cdot \ln k+\ln j)$ approximation on
$K_{i,j}$-free graphs (and an $\Oof(t\cdot \ln k)$ approximation on 
$K_{t,t}$-free graphs, respectively) as follows. 
For each of the sets $D_0, D_1,\ldots$ constructed in the course of 
the algorithm, run the standard greedy algorithm to extend it to a
dominating set, and return the smallest of the sets obtained in this way.
Letting~$p$ be the first index such that $D_p$ dominates all but at 
most $n=k^i\cdot (i+j)$ vertices of the graph, the above argument 
shows that $|D_p|\leq i\cdot k$. The standard greedy algorithm then 
adds at most $\ln n\cdot k \in \Oof((i\cdot \ln k + \ln j)\cdot k)$ further vertices 
to the dominating set, resulting in a dominating set of size 
$\Oof((i\cdot \ln k + \ln j)\cdot k)$.

\smallskip
We now modify the algorithm slightly to work without the parameter $i$. 
In each round let the algorithm choose elements $v_1,\ldots, 
v_{\ell}$, defining sets $B_1,\ldots, B_\ell$ in the above notation, until 
we do not find a vertex $v_{\ell+1}$ defining a set $B_{\ell+1}$ with 
$|B_{\ell+1}|\geq \ell+1$ any more. Let $t=\ell+1$ for the largest 
$\ell$ that was encountered in any round. Hence, the modified 
algorithm chooses at most $t-1$ elements in every round. 
Observe that in this 
construction, when we are at step $i$ and the corresponding 
set~$B_i$ has size at least $j$, $1\leq i\leq \ell$, $j\geq 1$, 
then we have found a subgraph $K_{i,j}$. 
Hence, $t$ is the least number such that the algorithm did not 
find $K_{t,t}$ as a subgraph and we can argue as above
that the algorithm performs as if $K_{t,t}$ was excluded from 
$G$. We output $K_{t-1,t-1}$ as a witness for this performance
guarantee.

\smallskip
Finally, note that the algorithm can be used to approximate the 
minimum size of a set which dominates a given subset $S$ of vertices 
of the graph, by initializing $A_0=S$ instead of $A_0=V(G)$.

\section{Hardness beyond degenerate graphs}

By the result of Bansal and Umboh~\cite{bansal2017tight} one
can compute a $3d$ approximation of a minimum dominating
set on any $d$-degenerate graph. The approximation factor
was improved to $2d+1$ by Dvo\v{r}\'ak~\cite{dvovrak2017distance}. 
To the best of our knowledge, degenerate
graphs are currently the most general graphs on which polynomial
time constant factor approximation algorithms for the dominating 
set problem are known. It is easy to see that the existence of such
algorithms on bi-clique free graphs, even on $K_{3,3}$-free graphs,
is unlikely. This result is a simple consequence of the following
result of Kumar et al.~\cite{kumar2000hardness}. Given a family $\Ff$ of 
subsets of a set~$A$, a \emph{set cover} is a 
subset $\Gg\subseteq \Ff$
such that $\bigcup_{F\in \Gg} F=A$. The 
\textsc{Minimum Set Cover} problem
is to find a minimum size set cover. The \emph{intersection}
of a set family $\Ff$ is the maximum size of the intersection
of two sets from $\Ff$. 

\begin{theorem}[Kumar et al.~\cite{kumar2000hardness}]
The {\normalfont\scshape Minimum Set Cover} problem on set families of intersection~$1$
cannot be approximated to within a factor of 
$c\frac{\log n}{\log\log n}$ for some constant $c$ in 
polynomial time unless for some constant 
$\epsilon< 1/2$ it holds that $\textsc{NP}\subseteq 
\textsc{DTIME}(2^{n^{1-\epsilon}})$. 
\end{theorem}

Now it is easy to derive the following theorem. 

\begin{theorem}\label{thm:hardness}
The {\normalfont\scshape Minimum Dominating Set} problem on $K_{3,3}$-free graphs 
cannot be approximated to within a factor of 
$c\frac{\log n}{\log\log n}$ for some constant $c$ in 
polynomial time unless for some constant 
$\epsilon< 1/2$ it holds that $\textsc{NP}\subseteq 
\mathrm{DTIME}(2^{n^{1-\epsilon}})$. 
\end{theorem}
\begin{proof}
We present an approximation preserving reduction from 
\textsc{Minimum Set Cover} on instances of intersection $1$ 
to \textsc{Minimum Dominating Set} on $K_{3,3}$-free graphs.
Let $\Ff$ be an instance of \textsc{Minimum Set Cover} with
intersection $1$. Let $A=\bigcup_{F\in\Ff}F$. 
We compute in polynomial time an instance
of \textsc{Minimum Dominating Set} on a graph $G$ as follows. We let 
$V(G)=A\cup \Ff\cup \{x,y\}$, where $x,y$ are new vertices that
do not appear in $A$. We add all edges $\{u,F\}$ if $u\in F$, as well as all edges $\{x,F\}$ for $F\in\Ff$ and the edge $\{x,y\}$. 

Now if $\Gg\subseteq \Ff$ is a feasible solution for the \textsc{Minimum Set Cover} instance, then $\Gg$ (as a subset of $G$) together with the
vertex $x$ is a dominating set for $G$ of size at most $|\Gg|+1$.
Conversely, let $D$ be 
a dominating set for $G$. We construct another dominating set $X$ 
such that $|X|\leq |D|$ and $X\subseteq \Ff\cup\{x\}$. We simply 
replace each $u\in A$ by a neighbour $F\in\Ff$. Furthermore, if
$y\in D$, we replace $y$ by $x$. Observe that $x$ or $y$ must
belong to $D$, as~$y$ must be dominated. Hence, in any case, 
$x\in X$. Now $\Gg\cap X$ is a set cover of size $|X|-1$. Hence, 
the reduction preserves approximations. 

Let us show that $G$ excludes $K_{3,3}$ as a subgraph. Assume
towards a contradiction that $K_{3,3}\subseteq G$. Then 
$K_{2,3}\subseteq G-x$. Since $G$ is bipartite we find 
elements $a_1,a_2\in A$ and $F_1,F_2\in \Ff$ (as vertices of $G$)
with $\{a_i, F_j\}\in E(G)$, $i,j\in\{1,2\}$, which form a $K_{2,2}$
subgraph of this graph. By construction of~$G$ we have 
$|F_1\cap F_2|\geq |\{a_1,a_2\}|=2$, contradicting that $\Ff$ is a 
\textsc{Minimum Set Cover}
instance with intersection~$1$. 

Finally observe that the reduction is 
obviously polynomial time computable. 

\mbox{}
\end{proof}

\paragraph{Acklowledgements}
I thank Saket Saurabh for pointing me to the work of Jones et 
al.~\cite{jones2013parameterized}. I thank the anonymous reviewers
for their valuable comments, and in particular for pointing out that
the modified greedy algorithm can be improved to an 
$\Oof(t\cdot \ln k)$ approximation by running the classical greedy
algorithm in parallel.

\bibliographystyle{abbrv}
\bibliography{ref} 

\end{document}